\documentclass[12pt]{article}
\usepackage{a4}
\usepackage{amsfonts,amssymb,amsmath}
\usepackage{graphicx}
\usepackage{cite}
\usepackage{ytableau}
\newtheorem{theorem}{Theorem}

\newtheorem{definition}[theorem]{Definition}

\newtheorem{proposition}[theorem]{Proposition}

\newenvironment{proof}[1][Proof]{\textbf{#1.} }{\ \rule{0.5em}{0.5em}}
\begin{document}

\title{Weak associativity and deformation quantization}

\vspace{1cm}
\author{V.G. Kupriyanov\thanks{vladislav.kupriyanov@gmail.com}\\
{\it CMCC-Universidade Federal do ABC, Santo Andr\'e, SP, Brazil}\\
{\it Tomsk State University, Tomsk, Russia}}

\maketitle

\begin{abstract}
Non-commutativity and non-associativity are quite natural in string theory. For open strings it appear due to the presence of non-vanishing background two-form in the world volume of Dirichlet brane, while in closed string theory the flux compactifications with non-vanishing three-form also lead to non-geometric backgrounds. In this paper, working in the framework of deformation quantization, we study the violation of associativity imposing the condition that the associator of three elements should vanish whenever each two of them are equal. The corresponding star products are called alternative and satisfy an important for physical applications properties like the Moufang identities, alternative identities, Artin's theorem, etc. The condition of alternativity is invariant under the gauge transformations, just like it happens in the associative case. The price to pay is the restriction on the non-associative algebra which can be represented by the alternative star product, it should satisfy the Malcev identity. The example of nontrivial Malcev algebra is the algebra of imaginary octonions. For this case we construct an explicit expression of the non-associative and alternative star product. We also discuss the quantization of Malcev-Poisson algebras of general form, study its properties and provide the lower order expression for the alternative star product. To conclude we define the integration on the algebra of the alternative star products and show that the integrated associator vanishes.
\end{abstract}

\newpage
\section{Introduction}\label{sec-intro}

In the canonical formulation of quantum mechanics the physical observables are represented by the hermitian linear operators $\hat f$ acting on the Hilbert space. The composition of quantum mechanical operators in general is non-commutative, $\hat f\hat g\neq\hat g\hat f$, yielding the uncertainty principle, but should be necessarilly associative, $\hat f(\hat g\hat h)=(\hat f\hat g)\hat h$. The later property implies two important identities involving the commutator of operators, the Leibniz rule, $[\hat f\hat g,\hat h]=[\hat f,\hat h]\hat g+\hat f[\hat g,\hat h]$, and the Jacobi Identity,
\begin{equation}\label{i1}
 [\hat f,[\hat g,\hat h]]+[\hat h,[\hat f,\hat g]]+[\hat g,[\hat h, \hat f]]\equiv0.
 \end{equation}
The time evolution in the Heisenberg picture is postulated by the operator equation
\begin{equation}\label{i2}
 i\hbar\frac{d}{dt}\hat f=i\hbar \frac{\partial\hat f}{\partial t}+[\hat H,\hat f],
 \end{equation}
 where $\hat H$ stands for the hamiltonian operator.
The Leibniz rule and the Jacobi identity guaranty the consistency of the quantum dynamics, meaning that the time evolution will preserve the algebra of physical observables. In particular, if $\hat f$ and $\hat g$ are two integrals of motion, $[\hat H,\hat f]=[\hat H,\hat g]=0$, the commutator $[\hat f,\hat g]$ is also an integral of motion, 
\begin{equation}\label{i3}
 i\hbar\frac{d}{dt}[\hat f,\hat g]=[\hat H,[\hat f,\hat g]]=0,
 \end{equation}
due to (\ref{i1}) and (\ref{i2}). So, the associativity is essential for the consistency of the canonical quantum mechanics.

However, some quantum mechanical systems are formulated in terms of non-associative algebras of the canonical operators. The standard exemple of such a situation is the introduction of the magnetic charges through the commutator of the covariant momenta, $[\hat x^i,\hat x^j]=0$, $[\hat x^i,\hat \pi_j]=i\delta^i_j$ and $[\hat \pi_i,\hat \pi_j]=ie\varepsilon_{ijk}B^k(\hat x)$, with $div \vec B\neq0$, see e.g., \cite{Jackiw, BaLu} for more details. For the Dirac monopole, in particular, one has $\vec B(\vec x)=g\vec x/ x^3$, with $g$ being the magnetic charge. The Jacobi identity is violated only in one point (the position of the magnetic charge). To overturn this difficulty one may impose the appropriate boundary condition for the wave function leading to the famous Dirac quantization rule: $eg/2\pi\hbar \in \mathbb{Z}$. For the linear magnetic field, $\vec B= g\vec x/3$, the jacobiator is constant, meaning that one cannot repeat the same trick to fix the problem.

Another source of the examples of non-associative structures is the string theory. Recent advances in understanding flux compactifications of string theory have suggested that non-geometric frames are related to non-commutative and non-associative deformations of space-time geometry \cite{Lust1,Lust2,Lust3,MSS2,Andriot}. Since, these flux deformations of geometry are probed by closed strings, they have a much better potential for providing an effective target space description of quantum gravity than previous appearances of non-commutative geometry in string theory.
To give an example of arising a non-geometric background let us consider the  closed strings propagating in a three-torus endowed with a constant Neveu-Schwarz flux $H = dB$. Applying consecutive $T$-duality transformations along all three directions one obtains the relation between geometric and non-geometric fluxes: $H\rightarrow f\rightarrow Q\rightarrow R$. The $Q$-flux background is non-commutative but associative, while the purely non-geometric $R$-flux background is not only non-commutative, but also non-associative. The presence of non-vanishing three-form $H$-flux in string compactifications makes the closed strings coordinates non-commutative and non-associative in the analogy with the coordinates of the open string endpoints attached to a $D$-brane in a background $B$-field \cite{ChuHo,Schom,CoSch,MK1,MK2}.

Mainly motivated by the string theory arguments there is a growing interest to the theories on non-geometric backgrounds. In particular, the non-associative quantum mechanics was studied in \cite{MSS1,Bojowald1,Bojowald2}. For recent developments in the area of non-associative field theory and non-associative gravity one may see \cite{GAS1,GAS2,Fuchs} and references therein.

The aim of this paper is to study the violation of the associativity in the framework of deformation quantization \cite{BFFLS}. In this approach to quantum mechanics the physical observables are represented by the functions on smooth manifold, $f\in C^{\infty}(\mathcal{M})$. To reflect the non-commutative nature of the composition low of quantum observables the ordinary point-wise multiplication of functions is substituted by the star multiplication, $f\cdot g\rightarrow f\star g, $ satisfying some natural restrictions described in the Sec. 2 and 3. In particular, star products representing different quantization prescriptions of the same classical system should be related by the gauge transformation, see Sec. 3 for details. 

To define the non-associative star products we are looking for the condition which would be gauge invariant and include the associative star products as a particular case. We imply the requirement that the star associator of any three elements should vanish whenever two of them are equal. The star product satisfying this condition is called alternative. Any associative multiplication is automatically alternative. We show that the proposed condition is invariant under the gauge transformations in a sense of a deformation quantization.

The alternative multiplications enjoy an important properties, like the Moufang identities, which can be used in physics. In particular, in non-associative quantum mechanics one may use the Moufang identities for the definition of states and uncertainty relations \cite{Bojowald1}. On the other hand the Moufang identities for the star product the star imply the Malcev identity on the commutator restricting the non-associative algebras which can be treated with the help of the alternative star product, see the discussion in the Sec. 4. 

On the classical level, the Malcev identity is equivalent to the Malcev-Poisson identity for the corresponding bracket. In Sec. 5 we discuss the classical dynamics on the Malcev-Poisson manifold. We introduce the modified bracket and show that there is a weaker analogue of the classical Poisson theorem, the modified bracket $\{f,g\}_H$ of two integrals of motion $f$ and $g$ is an integral of motion once one of them is manifestly time independent, $\partial f/\partial t=0$. So, the Malcev-Poisson identity can be useful to construct the integrals of motion.

The non-trivial example of the Malcev algebra is the algebra of the imaginary octonions, described in the Sec. 6. For this algebra we obtain an explicit formula for the alternative and non-associative star product. In the Sec. 7, we discuss the deformation quantization of the Malcev-Poisson structures of the general form. First, following our previous work \cite{starpr} we describe the construction of non-associative weakly-Hermitian Weyl star product. Then we prove some important properties of the alternative Weyl star products and deriva the lower order formula for it. Finally, in the Sec. 8, we discuss the definition of the trace functional (integration) on the algebra of the alternative star product and show that the integrated associator vanishes. 

\section{Alternative star product}

Star product is an important mathematical tool in theoretical physics with inumerous applications. Originally it appeared in the context of deformation quantization approach to quantum mechanics (QM) \cite{BFFLS} and is defined as a formal deformation of the usual point-wise multiplication of smooth functions on some manifold $\mathcal{M}$,
\begin{equation}
f\star g = f \cdot g +\sum_{r=1}^\infty (i\alpha)^r C_r(f,g) \,,\label{star}
\end{equation}
where $\alpha$ is a deformation parameter and $C_r$ denote some bi-linear bi-differential operators.  This product
provides a quantization of the bracket, 
\begin{equation}\label{m1}
   \{f,g\}=P^{ij}(x)\partial_i f\partial_j g,
\end{equation}
with $P^{ij}(x)$ being an antisymmetric bi-vector, if
\begin{equation}
 \left. \frac{f\star g-g \star f}{2i\alpha}\right|_{\alpha=0}= \{ f,g\}\,.\label{quant}
\end{equation}
Let us define the associator of three functions as
\begin{equation}
A(f,g,h)=f\star (g\star h)-(f\star g)\star h \,.\label{assoc}
\end{equation}
In the classical setting, the associativity of star products reflects the associativity of compositions of quantum mechanical operators and can be imposed by the equation,
\begin{equation}
A(f,g,h)=0 \,.\label{assoc1}
\end{equation}
Already in the second order, $\mathcal{O}(\alpha^2)$, this condition implies the Jacobi identity (JI) on the bi-vector field $P^{ij}(x)$,
\begin{eqnarray}\label{JI}
&&\{f,g,h\}=0,\\
&&\{f,g,h\}=\{f,\{g,h\}\}+\{h,\{f,g\}\}+\{g,\{h,f\}\},\nonumber
\end{eqnarray}
which is the consistency condition in this case. Also, (\ref{assoc}) allows to proceed to higher orders, $C_r(f,g)$, $r>1$. It happens that no other restrictions appear and on any Poisson manifold there exists an associative star product \cite{Kontsevich}.

Suppose now that the JI is violated, $\{f,g,h\}\neq0,$ and that the star product is no longer associative. However, let us keep the violation of the associativity under control.

\begin{definition} The star product is alternative if for any $f$ and $g$, one has,
\begin{eqnarray}
f\star(f\star g)=(f\star f)\star g,\label{a1} \\ f\star(g\star g)=(f\star g)\star g. \label{a2}
\end{eqnarray}
\end{definition}
Substituting in the above formulae $f$ by $f+h$ one obtains,
\begin{eqnarray*}
A(f,h,g)+A(h,f,g)=0, \\A(g,f,h)+A(g,h,f)=0, \end{eqnarray*}
which means that the associator $A(f,g,h)$ must be totally antisymmetric, i.e., is an alternating function of its arguments. That is why such an algebra is called alternative.
From, $A(f,g,f)=0$, follows the flexible identity,
\begin{equation}\label{a3}
 (f\star g)\star f=f\star(g\star f).
\end{equation}
The identities (\ref{a1}) and (\ref{a2}) are known as the left and the right alternative identities correspondingly. Each two of three identities (\ref{a1}-\ref{a3}) imply the third one. 

The star Jacobiator,
\begin{equation}\label{jac}
  [f,g,h]:=[f,[g,h]]+[h,[f,g]]+[g,[h,f]],
\end{equation}
where $[f,g]=f\star g-g\star f$, is antisymmetric by the definition, and can be written in terms of the associator (\ref{assoc}) as follows,
\begin{equation}\label{jass}
    [f,g,h]=A(f,g,h)-A(f,h,g)+A(h,f,g)-A(h,g,f)+A(g,h,f)-A(g,f,h).
\end{equation}
Since, for the alternative star product the associator should be totally antisymmetric, one finds
\begin{equation}\label{alt}
A(f,g,h)=  \frac{1}{6}  [f,g,h],\,\,\,\,\,\forall f,g,h.
\end{equation}
That is, the star product is alternative {\it if and only if} the relation (\ref{alt}) holds.

\section{Gauge invariance}

A milestone result in deformation quantization is the Formality Theorem by Kontsevich \cite{Kontsevich} that demonstrates the existence of associative star products and states that each two star products $\star$ and $\star^\prime$ corresponding to the same Poisson (satisfying JI) bi-vector $P^{ij}(x)$ are related by the gauge transformation
\begin{equation}
f\star^\prime g=\mathcal{D}^{-1}\left( \mathcal{D}f\star   \mathcal{D}g \right), \,\,\,\,\,\,\,\mathcal{D}=1+\mathcal{O}(\alpha)\,.\label{gauge}
\end{equation}
Gauge equivalence classes of star products are classified by Poisson structures on $\mathcal{M}$. Different star products providing quantization of the same bracket reflect different quantization prescriptions of the same classical model. Physical consequences of each choice of the quantization scheme may be different, however the basic principles of the consistency must be the same. That is why, it is important that the condition restricting the higher order terms of the star product (\ref{star}) should distinguish the whole class of admissible quantizations. 
 
 Suppose the star products $\star$ and $\star^\prime$ are related by the gauge transformation (\ref{gauge}). Let us calculate
\begin{eqnarray*}
A^\prime(f,g,h)&=&f\star^\prime (g\star^\prime h)-(f\star^\prime g)\star^\prime h\\
&=&\mathcal{D}^{-1}\left\{ \mathcal{D}f\star   \mathcal{D}\left[\mathcal{D}^{-1}\left( \mathcal{D}g\star   \mathcal{D}h \right)\right]\right\}-\mathcal{D}^{-1}\left\{ \mathcal{D}\left[\mathcal{D}^{-1}\left( \mathcal{D}f\star   \mathcal{D}g \right)\right]\star   \mathcal{D}h\right\}\,.\nonumber
\end{eqnarray*}
Since, the gauge operator $\mathcal{D}$ is linear, one finds
\begin{equation}
A^\prime(f,g,h)=\mathcal{D}^{-1}A (\mathcal{D}f,\mathcal{D}g,\mathcal{D}h) \,.\label{assoc2}
\end{equation}
For the Jacobiators, calculated with respect to the $\star$ and $\star^\prime$, one finds the similar relation,
\begin{equation}
 [f,g,h]^\prime=\mathcal{D}^{-1} [\mathcal{D}f,\mathcal{D}g,\mathcal{D}h] \,.\label{jac2}
\end{equation}
The combination of (\ref{assoc2}) and (\ref{jac2}) gives
\begin{equation}\label{alt1}
   [f,g,h]^\prime-6A^\prime(f,g,h)= \mathcal{D}^{-1}\left[ [\mathcal{D}f,\mathcal{D}g,\mathcal{D}h]-6A(\mathcal{D}f,\mathcal{D}g,\mathcal{D}h)\right]=0,
\end{equation}
because of the equation (\ref{alt}) and the invertibility of the operator $\mathcal{D}=1+\mathcal{O}(\alpha)$. This means that the condition (\ref{alt}) is invariant under the gauge transformations (\ref{gauge}). If the star product $\star$ is alternative, then any $\star^\prime$ gauge equivalent to $\star$ is also alternative. The situation is similar to the associative case. An important question here is whether exist a gauge transformation relating different alternative star products corresponding to the same antisymmetric bi-vector $P^{ij}(x)$, if not then how to classify gauge equivalence classes of alternative star products.

\section{Malcev identity} 

The first point we would like to address here is whenever the alternative star product exists for each non-Poisson bracket (\ref{quant}), or the equation (\ref{alt}) implies some restrictions on the bi-vector $P^{ij}(x)$. 
In every alternative algebra the following identities are valid, see \cite{Ivan-book}:
\begin{eqnarray}
f\star(h\star g\star h)=((f\star h)\star g)\star h, \\ (h\star f\star h)\star g=h\star (f\star(h\star g)), \nonumber\\ (h\star g)\star(f\star h)=h\star (g\star f)\star h, \nonumber
\end{eqnarray}
known as the right, the left and the middle Moufang identities. The left and the right Moufang identities are equivalent respectively to
\begin{eqnarray}
A(h,f,h\star g)=A(h,f,g)\star h,\label{mo1} \\ A(h,f,g\star h)=h\star A(h,f,g).\label{mo2}
\end{eqnarray}
Indeed,
\begin{eqnarray*}
&&A(f,h\star g,h)+A(f,h,g)\star h=f\star(h\star g\star h)-((f\star h)\star g)\star h, \\&& A(h,g\star h,f)+h\star A(g,h,f)=-(h\star g \star h)\star f+h\star (g\star(h\star f)).
\end{eqnarray*}
Now subtracting (\ref{mo2}) from (\ref{mo1}) and using (\ref{alt}) we end up with the Malcev identity:
\begin{equation}\label{Malcev}
[h,f,[h,g]]=[[h,f,g],h].
\end{equation}
That is, if the star product is alternative, the star commutator should necessarily satisfy the Malcev identity (\ref{Malcev}). In the same manner as in the associative case the equation (\ref{assoc1}) implies the JI on the classical Poisson bracket, $\{f,g\}_{Pb}$, the requirement (\ref{alt}) restricts the classical bracket, (\ref{m1}). For any three functions $f$, $g$ and $h$ it should satisfy,
\begin{equation}\label{m3}
  \{h,f,\{h,g\}\}=\{\{h,f,g\},h\}.
\end{equation} 
The anti-symmetric bilinear bracket, satisfying (\ref{m3}) and the Leibniz rule is known as the Malcev-Poisson bracket \cite{Ivan}.

Any Poisson bracket, $\{f,g\}_{Pb}$, obviously satisfies the identity (\ref{m3}). The example of brackets for which the JI is violated while (\ref{m3}) holds is discussed in the Sec. 6, it is the bracket representing the algebra of imaginary octonions. To give an example of bracket for which the Malcev-Poisson identity (\ref{m3}) is violated let us consider the algebra corresponding to the Dirac monopole. One may see that in this case the Malcev-Poisson identity is violated already on the coordinate functions $x^i$, $x^j$, and $x^k$, see \cite{Gunaydin}. 

If $A$ is an alternative algebra and $ab$ is a multiplication in $A$, then the commutator, $[a,b]=ab-ba$, $ a,b\in A$, is antisymmetric and satisfies the Malcev identity, thus it defines the Malcev algebra $M$. This relation between an alternative algebra and a Malcev algebra is well known in mathematics. Whenever the inverse statement is true, i.e., if for each Malcev algebra $M$ there exist a corresponding alternative algebra $A$, such that the multiplication $[a,b]$ in $M$ can be represented as a commutator $ab-ba$ in $A$, is an open problem. In the context of deformation quantization this problem can be formulated as the existence of the alternative star product representing quantization of a given Malcev-Poisson bracket. 

\section{Dynamics on Malcev-Poisson manifolds}

In the classical mechanics the quantum commutator $[\hat f,\hat g]$ corresponds to the Poisson bracket $\{f,g\}_{Pb}$ which is antisymmetric, bi-linear and should satisfy the Leibniz rule and the Jacobi identity. The Poisson bracket gives rise to a Hamilton's equations of motion,
\begin{equation*}
\frac{df}{dt} = \frac{\partial f}{\partial t}+\{ H, f\}_{Pb},
 \end{equation*}
with $H$ being a Hamiltonian function. The importance of the Jacobi Identity for the Poisson bracket is that together with the Leibniz rule it guaranties that the corresponding time evolution preserves the algebra (with respect to the Poisson bracket) of the physical observables,
\begin{equation*}
\frac{d}{dt} \{f,g\}_{pb}=  \{\dot f,g\}_{Pb}+\{f,\dot g\}_{Pb}.
\end{equation*}
 A Poisson bracket of two integrals of motion is again an integral of motion. This statement is also known as a classical Poisson theorem.

However, in some cases the classical dynamics of the physical system is formulated in terms of non-hamiltonian equations of motion involving the bracket (\ref{m1}) which does not satisfy the Jacobi identity. The problem here is that the observables do not form an algebra with respect to the new bracket. The new bracket of two integrals of motion is no longer an integral of motion. Our aim here is to study to which end the Malcev identity (\ref{m3}) can be useful in the investigation of the classical dynamics of the system. In particular, in the construction of the integrals of motion.

The time evolution governed by the hamiltonian $H$ is defined by the equations
\begin{equation}\label{m5}
 \frac{df}{dt}=\frac{\partial f}{\partial t}+\{H,f\}.
\end{equation}
Let us suppose that the functions $f(x)$ and $g(x)$ do not depend manifestly on time and are integrals of motion, i.e.,
\begin{equation}\label{m7}
 \{f,H\}=\{g,H\}=0.
\end{equation}
If the Jacobian $\{H,f,g\}=0$, then just like in the standard Poisson case the bracket $\{f,g\}$, is an integral of motion,
\begin{equation}\label{m5a}
 \frac{d}{dt}\{f,g\}=\{H,\{f,g\}\}=\{H,f,g\}-\{g,\{H,f\}\}-\{f,\{g,H\}\}=0.
\end{equation}
If the Jacobian $\{H,f,g\}\neq0$, then due to the Malcev identity it is an integral of motion,
\begin{equation}\label{m5a}
 \frac{d}{dt}\{H,f,g\}=\{H,\{H,f,g\}\}=-\{H,f,\{H,g\}\}=0.
\end{equation}
Moreover, the quantity
\begin{equation}\label{m4}
  \{f,g\}_H=\{f,g\}-\{H,f,g\}\ t,
\end{equation}
where $t$ stands for the time variable, $\{t,f\}=0$ is also an integral of motion,
\begin{eqnarray}\label{m8}
 \frac{d}{dt} \{f,g\}_H&=&\frac{\partial }{\partial t}\{f,g\}_H+\{H,\{f,g\}_H\}\\&=&-\{H,f,g\}+\{H,\{f,g\}\}-\{H,\{H,f,g\}\}t \nonumber\\
 &=&\{H,f,\{H,g\}\}t=0.\nonumber
\end{eqnarray}
Since, $\{H,f,g\}$ does not depend manifestly on time, while $ \{f,g\}_H$ is linear in time by definition, these two quantities are functionally independent. In principle, it may be helpful to study the classical dynamics.

In general, one may see that
\begin{equation}\label{brk}
\frac{d}{dt} \{f,g\}_H=  \{\dot f,g\}_H+\{f,\dot g\}_H-\{H,g,\{H,f\}\}t,
\end{equation}
So, the modified bracket $ \{f,g\}_H$ of two integrals of motion $f$ or $g$ is an integral of motion only if one of them is manifestly time independent, i.e., $\{H,f\}=-\partial f/\partial t=0$.

\section{Star product for octonions}

 Probably, the most known example of the alternative, but non-associative algebra are the octonions. Every octonion $\mathcal{X}$ can be written in the form
\begin{equation}\label{oct}
\mathcal{X}=\mathcal{X}_0e_0+\mathcal{X}_1e_1+\mathcal{X}_2e_2+\mathcal{X}_3e_3+\mathcal{X}_4e_4+\mathcal{X}_5e_5+\mathcal{X}_6e_6+\mathcal{X}_7e_7,
\end{equation}
where $\mathcal{X}_i$ are the real coefficients, $e_0$ is the scalar element, and the imaginary unit octonions $e_i$ satisfy the following multiplication rule:
\begin{equation}\label{oct1}
e_ie_j=-\delta_{ij}e_0+\varepsilon_{ijk}e_k,
\end{equation}
with $\varepsilon_{ijk}$ being a completely antisymmetric tensor of the rank three with a positive value $+1$ when $ijk = 123,$ $145,$ $176,$ $246,$ $257,$ $347,$ $365$, and $i,j,k=1\dots7$. The octonions are neither commutative, nor associative. The commutator algebra of the octonions is given by
\begin{equation}\label{oct2}
[e_i,e_j]:=e_ie_j-e_je_i=2\varepsilon_{ijk}e_k.
\end{equation}
In seven dimensions one has
\begin{equation}\label{epsilon7}
\varepsilon_{ijk}\varepsilon_{lmk}=\delta_{il}\delta_{jm}-\delta_{im}\delta_{jl}+\varepsilon_{ijlm},
\end{equation}
where $\varepsilon_{ijlm}$ is a completely antisymmetric tensor of the rank four with a positive value $+1$ when $ijlm = 1247,$ $1265,$ $1436,$ etc., One may also represent the tensor with four indices $\varepsilon_{ijlm}$ as a dual to the tensor with three indices $\varepsilon_{kpr}$ through 
\begin{equation}\label{epsilon8}
\varepsilon_{ijlm}=\frac{1}{6}\varepsilon_{ijlmkpr}\varepsilon_{kpr},
\end{equation}
where $\varepsilon_{ijlmkpr}$ is the Levi-Civita symbol in seven dimensions, normalized as $\varepsilon_{1234567}=+1$. Taking into account (\ref{epsilon7}), for the Jacobian we get
\begin{equation}\label{oct3}
[e_i,[e_j,e_k]]+cycl(ijk)=6\varepsilon_{ijkl}e_l.
\end{equation}
In this section we are interested in the quantization of the bracket:
\begin{equation}\label{oct4}
\{x^i,x^j\}=\varepsilon^{ijk}x^k, \,\,\,\,\,i,j,k=1\dots7,
\end{equation}
representing the algebra of imaginary octonions.

\subsection{Vector star product} Consider the following multiplication map, for each two vectors $\vec p_1,\vec p_2$ from the unit ball, $p_ip^i\leq1$, it is assigned the vector $\vec p_1\circledast\vec p_2$ by the rule,
\begin{equation}\label{vstar}
\vec p_1\circledast\vec p_2=\varepsilon(\vec p_1,\vec p_2)\left({\sqrt{1-| \vec p_1|^2} }\,\,\vec p_2+\sqrt{1-| \vec p_2|^2}\,\,\vec p_1-\vec p_1\times\vec p_2\right),
\end{equation}
    where $| \vec p|=\sqrt{ \vec p\cdot \vec p}$ is the Euclidean vector norm, $\vec p_1\cdot \vec p_2$, $\vec p_1\times\vec p_2$ stand for the dot and the cross products correspondingly 
    and $\varepsilon(\vec p_1,\vec p_2)=\pm1$, is the sign of the expression $\sqrt{1-|\vec  p_1|^2}\sqrt{1-|\vec p_2|^2} -\vec p_1\cdot\vec p_2$. Since,
    \begin{equation*}
1-|\vec p_1\circledast\vec p_2|^2=\left(\sqrt{1-|\vec  p_1|^2}\sqrt{1-|\vec p_2|^2}-\vec p_1\cdot\vec p_2\right)^2\geq0,
\end{equation*}
the resulting vector $\vec p_1\circledast\vec p_2$ also belongs to the unit ball. This multiplication is non-commutative, but as we will see below is alternative.

Taking into account the property of the mixed product, $\vec p_1\cdot(\vec p_2\times\vec p_3)=\vec p_3\cdot(\vec p_1\times\vec p_2)$, one may fulfill the straitfoward calculation and see that,
 \begin{equation}\label{sign}
\varepsilon(\vec p_1,\vec p_2)\varepsilon(\vec p_1\circledast\vec p_2,\vec p_3)=\varepsilon(\vec p_1,\vec p_2\circledast\vec p_3)\varepsilon(\vec p_2,\vec p_3).
\end{equation}
With the help of the above formula for the associator of three vectors one finds
\begin{eqnarray}\label{assvstar}
&&A(\vec p_1,\vec p_2,\vec p_3)=\left(\vec p_1\circledast\vec p_2\right)\circledast\vec p_3-\vec p_1\circledast\left(\vec p_2\circledast\vec p_3\right)\\
&&=\varepsilon(\vec p_1,\vec p_2)\varepsilon(\vec p_1\circledast\vec p_2,\vec p_3)\left[\left(\vec p_1\times\vec p_2\right)\times\vec p_3-\left(\vec p_1\cdot\vec p_2\right)\vec p_3-\vec p_1\times\left(\vec p_2\times\vec p_3\right)+\left(\vec p_2\cdot\vec p_3\right)\vec p_1\right].\nonumber
\end{eqnarray}
The coordinates of the vector $\vec p_1\times\vec p_2$ can be written as
\begin{equation}
(\vec p_1\times\vec p_2)_i=\varepsilon_{ijk}p_1^jp_2^k,
\end{equation}
where $\varepsilon_{ijk}$ is defined after the eq. (\ref{oct1}). Taking into account (\ref{epsilon7}) one obtains for the associator (\ref{assvstar}) written in the components,
\begin{equation}
A(\vec p_1,\vec p_2,\vec p_3)_i=\pm2\varepsilon_{ijlm}p_1^jp_2^lp_3^m.
\end{equation}
It is different from zero, but is totally antisymmetric, meaning that in seven dimensions the vector star product (\ref{vstar}) is non-associative, but alternative.

To generalize the vector star product (\ref{vstar}) for the whole Euclidean space $V$ we introduce the map
\begin{equation}
 \vec p=\frac{\sin\left(\frac{\alpha}{2}|\vec k|\right)}{|\vec k|}\vec k,\,\,\,\,k^i\in\mathbb{R}.
\end{equation}
The inverse map is given by
\begin{equation}
 \vec k=\frac{2\arcsin\left|\vec p\right|}{\alpha\left|\vec p\right|}\vec p.
\end{equation}
For each two vectors $\vec k_1,\vec k_2\in V$ we set:
\begin{equation}
\label{Bk}
\mathcal{ B}(\vec k_1,\vec k_2)=\left.\frac{2\arcsin\left|\vec p_1\circledast\vec p_2\right|}{\alpha\left|\vec p_1\circledast\vec p_2\right|}\vec p_1\circledast\vec p_2\right|_{ \vec p_a=\vec k_a\sin\left(\frac{\alpha}{2}|k_a|\right)/|k_a|},\,\,\,a=1,2.
\end{equation}
From the definition (\ref{Bk}) one immediately finds the following properties:
\begin{description}
\item[i)]  $\mathcal{ B}(\vec k_1,\vec k_2)=-\mathcal{ B}(-\vec k_2,-\vec k_1)$;
\item[ii)] $\mathcal{ B}(\vec k,0)=\mathcal{ B}(0,\vec k)=\vec k$;
\item[iii)]  $\mathcal{ B}(\vec k_1,\vec k_2)=\vec k_1+\vec k_2-{\alpha}\vec k_1\times\vec k_2+\dots$;
\item[iv)] the associator, $$\mathcal{A}(k_1,k_2,k_3)=\mathcal{ B}(\mathcal{ B}(\vec k_1,\vec k_2),\vec k_3)-\mathcal{ B}(\vec k_1,\mathcal{ B}(\vec k_2,\vec k_3)),$$ is antisymmetric in all arguments.
\end{description}
The derivation of the vector star product, $\vec p_1\circledast\vec p_2$, and the expression for $\mathcal{ B}(\vec k_1,\vec k_2)$ in three space dimensions and its relation with BCH formula for $SU(2)$ can be found in \cite{FL,FM,OR13}.

\subsection{Star product} Let us define the star product as
\begin{equation}\label{w1}
f\star g(x) =\int \frac{d^{7}k_1}{\left( 2\pi
\right) ^{7}} \frac{d^{7}k_2}{\left( 2\pi
\right) ^{7}}\tilde{f}\left( k_1\right)\tilde{g}\left( k_2\right)e^{i\mathcal{ B}(\vec k_1,\vec k_2)\cdot{\vec x}},
\end{equation}
where $\tilde{f}$ stands for the Fourier transform of the function $f$. Due to the properties of the vector multiplication $\mathcal{ B}(\vec k_1,\vec k_2)$, the introduced star product is hermitian,
\begin{equation}\label{41}
    (f\star g)^\ast=g^\ast \star f^\ast,
\end{equation}
satisfies the stability of the unity, i.e., $f\star 1=1\star f=f(x)$. It provides the quantization of the bracket (\ref{oct4}),
\begin{equation}\label{41}
       \left. \frac{f\star g-g \star f}{2i\alpha}\right|_{\alpha=0}=x^i\varepsilon_{ijk}\partial_jf\partial_kg.
\end{equation}
The property ($\bf{iv}$) implies that the star product (\ref{w1}) is alternative. Consequently, the star commutator $[f,g]$ satisfy the Malcev identity (\ref{Malcev}).

Let us calculate $x^i\star f$ using (\ref{w1}). Since the Fourier transform of $x^i$ is the derivative of a Dirac delta function, we get
\begin{eqnarray}
x^i\star f &=&\int \frac{d^{7}k_1}{\left( 2\pi
\right) ^{7}} \frac{d^{7}k_2}{\left( 2\pi
\right) ^{7}}(2\pi i )^7\left(\partial^i_{k_1}\delta(k_1)\right)\tilde{f}\left( k_2\right)e^{i\mathcal{ B}(\vec k_1,\vec k_2)\cdot{\vec x}}\\
&=&-\int \frac{d^{7}k_1}{\left( 2\pi
\right) ^{7}} d^{7}k_2 x^l\frac{ \partial  \mathcal{ B}_l(\vec k_1,\vec k_2)}{\partial k_1^i}\delta(k_1)\tilde{f}\left( k_2\right)e^{i\mathcal{ B}(\vec k_1,\vec k_2)\cdot{\vec x}}\nonumber\\
&=&-\int \frac{d^{7}k_2}{\left( 2\pi
\right) ^{7}}  x^l\left.\frac{ \partial  \mathcal{ B}_l(\vec k_1,\vec k_2)}{\partial k_1^i}\right|_{k_1=0}\tilde{f}\left( k_2\right)e^{i\mathcal{ B}(0,\vec k_2)\cdot{\vec x}}\nonumber
\end{eqnarray}
After some algebra one may see that
\begin{eqnarray*}
&&\left.\frac{ \partial \mathcal{ B}_l(\vec k_1,\vec k_2)}{\partial k_1^i}\right|_{k_1=0}=\\
&&-\alpha \varepsilon^{ilm}k_2^m+\delta^{il}\frac{\alpha}{2}|k_2|\cot\left(\frac{\alpha}{2}|k_2|\right)
+\frac{k_2^ik_2^l}{|k_2|^2}\left(\frac{\alpha}{2}|k_2|\cot\left(\frac{\alpha}{2}|k_2|\right)-1\right).\nonumber
\end{eqnarray*}
Then, taking into account ({\bf ii}) and integrating over $k_2$ we conclude that
\begin{eqnarray}
x^i\star f &=&\left\{x^i+\frac{i\alpha}{2}\varepsilon^{ijk}x^k\partial_j\right.\\
&+&\left.
  (x^i\Delta-x^l\partial_l\partial_i)\Delta^{-1}\left[\frac{\alpha}{2}\sqrt{\Delta}\coth\left(\frac{\alpha}{2}\sqrt{\Delta}\right)-1\right]\right\}\triangleright f,\nonumber \label{poly}
\end{eqnarray}
where $\Delta=\partial_i\partial^i$.

For the Jacobiator one finds,
\begin{equation}\label{oct5}
[x^i,x^j,x^k]_\star=-3\alpha^2\varepsilon^{ijkl}x^l,
\end{equation}
which is in agreement with (\ref{oct3}).

\subsection{Integration} To introduce the trace functional (integration) on the algebra of the star product first we note that $  \partial_{i}\left( \varepsilon^{ijk}x^{k}  \right)  =0$, due to the antisymmetry of $\varepsilon^{ijk}$. Thus, for functions $f$ and $g$ vanishing on the infinity and the bracket (\ref{oct4}) one has, $\int  d^{7}x \{f,g\}=0$. That is, it is reasonable to take the canonical integration as the corresponding trace functional. However, the star product (\ref{w1}) is not closed,
 \begin{equation}\label{oct6}
 \int  d^{7}x [f\star g- f\cdot g]\neq0 .\end{equation}
To overcome this difficulty one may search for the gauge equivalent, in the sense (\ref{gauge}), star product which would be closed with respect to the introduced integration. In \cite{KV15} we did it for the star product corresponding to the algebra $\mathfrak{su}(2)$. The procedure of finding the gauge operator for the star product (\ref{w1}) is absolutely analogous to \cite{KV15}, except for some coefficients because of the different dimensions. Here we give only the final result. The star product
\begin{equation}\label{oct7}
f\circ g= \mathcal{D}^{-1}\left( \mathcal{D}f\star   \mathcal{D}g \right),
\end{equation}
where
\begin{equation}\label{oct8}
\mathcal{D}=\left( \frac{2\sinh\left(\frac{1}{2}\alpha\sqrt{\Delta}\right)}{\alpha\sqrt{\Delta}}\right)^{\frac{1}{3}},
\end{equation}
is closed with respect to the integration,
 \begin{equation}\label{oct9}
 \int  d^{7}x f\circ g= \int  d^{7}x f\cdot g .
 \end{equation}
By the construction, (\ref{oct7}) is alternative. Consequently, for this star product the integrated associator vanishes.

It is reasonable to ask here whether the non-geometric background described in this section can be obtained as a $T$-dualization of some topological background. Since we are talking about the star product corresponding to the algebra of $7$-dimensional imaginary octonions which in turn can be represented as a tangent space of the identity of the $7$-sphere, it is interesting to study if a $T$-dualization of a $7$-sphere with a constant three-form flux will result in the non-geometric space defined by (\ref{oct7}).

\section{Weyl star product}

In this section we discuss the nonassociative star products corresponding to the arbitrary (non)-Poisson bi-vector $P^{ij}(x)$. This approach is closer to the ideology
of quantum mechanics. We are able to construct
this product  for a completely general bi-vector $P^{ij}$. Besides, as it was shown in \cite{starpr},
the Weyl products have considerable computations advantages over other products.

Through the star product one may associate a (formal) differential operator $\hat f$ to a function $f$ as
\begin{equation}
(f\star g)(x)=\hat f \triangleright g(x)~,
\label{d3}
\end{equation}%
where the symbol $\triangleright$ on the right hand side means an action of a differential operator on
a function. In particular,
\begin{equation}\label{op1}
x^I\star f=\hat x^I\triangleright f(x)~,
\end{equation}%
To define the star product between two arbitrary functions, $f\star g$, let us introduce the notion of Weyl star product. If for any $f$ the operator $\hat f$ 
can be obtained by the Weyl symmetric ordering of operators $\hat x^j$, we call the start product (\ref{d3}) as the \emph{Weyl star product}. If $\tilde f(\xi)$ is a Fourier transform of $f$, then
\begin{equation}
\hat f=\hat f\left( \hat{x}\right) =W\left( f\right)  =\int \frac{d^{N}\xi}{\left( 2\pi \right) ^{N}}%
\tilde{f}\left( \xi\right) e^{-i\xi_{m}\hat{x}^{m}}.  \label{2}
\end{equation}
For example, $W(x^ix^j)=\tfrac 12 (\hat x^i\hat x^j + \hat x^j \hat x^i)$. 
Weyl star products satisfy
   \begin{equation}\label{weyl}
    (x^{i_1}\dots x^{i_n})\star f=\sum_{P_n} \frac 1{n!} P_n( x^{i_1}\star(\dots \star (x^{i_n}\star f)\dots)\,.
\end{equation}
where $P_n$ denotes a permutation of $n$ elements. It should be stressed that the correspondence $f\to \hat f$ is not an algebra representation. Since the star product that we
consider here does not need to be associative, in general $\widehat{ (f\star g)} \ne \hat f \circ \hat g$.

We will also need the following definitions. We call the star product {\it weakly Hermitean} if for all coordinates $x^j$,   
\begin{equation}\label{wher}
    (x^j\star f)^\ast=f^\ast \star x^j.
\end{equation}
The star product is {\it strictly triangular}, if the r.h.s. of
\begin{equation}\label{ST}
\frac{x^i\star x^j-x^j \star x^i}{2i\alpha}=P^{ij}(x),
\end{equation}
does not depend on $\alpha$.
The following condition,
\begin{equation}\label{su}
f(x)\star 1=1\star f(x)=f(x),
\end{equation}
is called the stability of unity. The conditions (\ref{weyl}) and (\ref{su}) together imply
  \begin{equation}\label{weyl1}
  \sum_{P_n} \frac 1{n!} P_n( x^{i_1}\star(\dots \star x^{i_n})\dots)=x^{i_1}\dots x^{i_n}\,.
\end{equation}

Now we are ready to formulate the following
\begin{theorem}
For any bi-vector field $P^{ij}(x)$ there exists unique weakly Hermitian strictly triangular Weyl star product satisfying the stability of
unity condition.
\end{theorem}
The weak Hermiticity and the stability of unity are physical requirements. The Weyl condition reflects the particular choice of the quantization prescription. While, the requirement of strict triangularity (\ref{ST}) is necessary to provide the uniqueness of the star product. The proof of this statement is constructive, see \cite{starpr}. We present recursion relations that allow to compute this star product to any given order. In particular, up to the third order we have:
\begin{eqnarray}
&&(f\star g)(x)=f\cdot g+i\alpha P ^{ij}\partial
_{i}f\partial _{j}g  \label{starw} \\
&&-\frac{\alpha ^{2}}{2}P^{ij}P ^{kl}\partial _{i}\partial
_{k}f\partial _{j}\partial _{l}g-\frac{\alpha ^{2}}{3}P ^{ij}\partial
_{j}P ^{kl}\left( \partial _{i}\partial _{k}f\partial _{l}g-\partial
_{k}f\partial _{i}\partial _{l}g\right) \nonumber \\
&&-i\alpha^3 \left[
\frac{1}{3}P ^{nl}\partial _{l}P ^{mk}\partial _{n}\partial
_{m}P ^{ij}\left( \partial _{i}f\partial _{j}\partial _{k}g-\partial
_{i}g\partial _{j}\partial _{k}f\right) \right.  \nonumber \\
&&+ \frac{1}{6}  P^{nk}\partial_nP^{jm}\partial_mP^{il}
\left(\partial_i\partial_jf\partial_k\partial_lg-\partial_i\partial_jg\partial_k\partial_lf\right)   \nonumber \\
&&+\frac{1}{3}P ^{ln}\partial _{l}P ^{jm}P ^{ik}\left(
\partial _{i}\partial _{j}f\partial _{k}\partial _{n}\partial _{m}g-\partial
_{i}\partial _{j}g\partial _{k}\partial _{n}\partial _{m}f\right)
\nonumber \\
&&+\frac{1}{6}P ^{jl}P ^{im}P ^{kn}\partial _{i}\partial
_{j}\partial _{k}f\partial _{l}\partial _{n}\partial _{m}g  \nonumber \\
&&\left. +\frac{1}{6}P ^{nk}P ^{ml}\partial _{n}\partial _{m}P
^{ij}\left( \partial _{i}f\partial _{j}\partial _{k}\partial _{l}g-\partial
_{i}g\partial _{j}\partial _{k}\partial _{l}f\right)\right] +\mathcal{O}\left( \alpha ^{4}\right)~.  \nonumber
\end{eqnarray}
Up to this order the star product is hermitian.  For linear bi-vectors, $P^{ij}(x)=C^{ij}_kx^k$, the weakly hermitian Weyl star product is also hermitian in all orders.
 Also one may check that if $P^{ij}(x)$ is a Malcev-Poisson structure, the conditions (\ref{a1}) and (\ref{a2}) are satisfied, i.e., the star product is alternative.

\subsection{Properties of the Weyl star product}

\begin{proposition} Weakly Hermitean Weyl star product is weakly alternative, i.e, satisfies
\begin{equation}\label{j5}
 x^i\star(x^j\star x^k)-(x^i\star x^j)\star x^k= \frac{1}{6}[x^i,x^j,x^k]_\star
\end{equation}
 \end{proposition}
\begin{proof} For the Weyl star product one has by (\ref{weyl})
\begin{equation}\label{xxx}
  (x^ix^j)\star x^k=\frac{1}{2}\left(x^i\star(x^j\star x^k)+x^j\star(x^i\star x^k)\right).
\end{equation}
On the other hand the stability of the unity,
\begin{equation}
  \frac{1}{2}\left(x^i\star x^j+x^j\star x^i\right)=x^ix^j,
\end{equation}
implies
\begin{equation}\label{j1}
 \left(x^i\star x^j\right)\star x^k+\left(x^j\star x^i\right)\star x^k=  x^i\star(x^j\star x^k)+x^j\star(x^i\star x^k).
\end{equation}
Which means that
\begin{equation}\label{j2}
  A(x^i,x^j,x^k)+ A(x^j, x^i,x^k)=0,
\end{equation}
that is, the associator of coordinate functions $x^i$, $x^j$ and $x^k$ is antisymmetric in first two arguments. If the Weyl star product is weak-Hermitian, then considering complex conjugate of the equation (\ref{j1}) and 
using (\ref{wher}) we obtain
\begin{equation}\label{j3}
  x^k\star(x^j\star x^i)+x^k\star(x^i\star x^j)= \left(x^k\star x^j\right)\star x^i+\left(x^k\star x^i\right)\star x^j.
\end{equation}
That is, the associator of coordinate functions $x^k$, $x^j$ and $x^i$ is antisymmetric in the last two arguments,
\begin{equation}\label{j4}
  A(x^k,x^j,x^i)+ A(x^k, x^i,x^j)=0.
\end{equation}
Using now (\ref{j2}) and (\ref{j4}) in (\ref{jass}) we find (\ref{j5}). \end{proof}

The same way that the weak Hermiticity does not necessarily implies 
Hermiticity, the weakly alternative star product should not necessarily be alternative.  

\begin{proposition}  For Weyl star product satisfying the stability of unity for any $f$ holds:
\begin{equation}\label{j6}
A(x^{i_1}, x^{i_2}\dots x^{i_n},f)+A(x^{i_2}, x^{i_3}\dots x^{i_1},f)+\dots+ A(x^{i_n},x^{i_1}\dots x^{i_{n-1}},f)=0\,.
\end{equation}
\end{proposition}
\begin{proof} From one side, from (\ref{weyl}) one finds,
  \begin{eqnarray}\label{j7}
&&  n  (x^{i_1}\dots x^{i_n})\star f=\\&& x^{i_1}\star((x^{i_2}\dots x^{i_n})\star f)+x^{i_2}\star((x^{i_3}\dots x^{i_1})\star f)+\dots+ x^{i_n}\star ((x^{i_1}\dots x^{i_{n-1}})\star f)\,.\nonumber
\end{eqnarray}
At the same time, since
  \begin{equation}\label{j8}
  n  x^{i_1}\dots x^{i_n}=x^{i_1}\star(x^{i_2}\dots x^{i_n})+x^{i_2}\star(x^{i_3}\dots x^{i_1})+\dots+ x^{i_n}\star (x^{i_1}\dots x^{i_{n-1}})\,,
\end{equation}
due to (\ref{weyl1}), one has,
  \begin{eqnarray}\label{j9}
&&  n  (x^{i_1}\dots x^{i_n})\star f=\\&&(x^{i_1}\star(x^{i_2}\dots x^{i_n}))\star f+(x^{i_2}\star(x^{i_3}\dots x^{i_1}))\star f+\dots+ (x^{i_n}\star (x^{i_1}\dots x^{i_{n-1}}))\star f\,.\nonumber
\end{eqnarray}
Subtracting (\ref{j9}) from (\ref{j7}) one obtaines (\ref{j6}). \end{proof}

\begin{proposition} The alternative weakly Hermitean Weyl star product is Hermitean. \end{proposition}
\begin{proof} The weakly Hermitean Weyl star product is Hermitean if 
  \begin{equation}\label{j10}
    \sum_{P_n} \frac 1{n!} P_n(\dots(f\star x^{i_n})\star\dots)\star x^{i_1})=f\star  (x^{i_1}\dots x^{i_n})\,.
\end{equation}
Let us prove by the induction that if the star product is also alternative then (\ref{j10}) holds true. For $n=2$,
\begin{eqnarray*}
&&(f\star x^i)\star x^j+(f\star x^j)\star x^i=\\&&f\star(x^i\star x^j)+f\star(x^j\star x^i)-A(f,x^i,x^j)-A(f,x^j,x^i)=2f\star(x^ix^j).
\end{eqnarray*}
Suppose that (\ref{j10}) holds for $n=k$ and let us prove it for $n=k+1$. We write,
 \begin{eqnarray*}
   && \sum_{P_{k+1}} \frac 1{(k+1)!} P_{k+1}(\dots(f\star x^{i_{k+1}})\star\dots)\star x^{i_1})=\\&&
   \frac{1}{k+1}\left[(f\star  (x^{i_1}\dots x^{i_k}))\star x^{i_{k+1}}+\dots+(f\star  (x^{i_{k+1}}\dots x^{i_{k-1}}))\star x^{i_{k}}\right]=\\&&
     \frac{1}{k+1}\left[f\star(  (x^{i_1}\dots x^{i_k})\star x^{i_{k+1}})+\dots+f\star ( (x^{i_{k+1}}\dots x^{i_{k-1}})\star x^{i_{k}})\right.\\&&
     -A(f,x^{i_1}\dots x^{i_k},x^{i_{k+1}})-\dots-A(f,x^{i_{k+1}}\dots x^{i_{k-1}},x^{i_{k}})]=
     f\star(x^{i_1}\dots x^{i_{k+1}}) \,.
\end{eqnarray*}
In the last line we have used the equation (\ref{j6}) and antisymmetry of the associator. \end{proof}

As it was already mentioned the main advantage to work with weak Hermitean Weyl star product is the existence \cite{starpr} of an explicit recursive formulae for its construction to any desired order for the arbitrary antisymmetric bi-vector $P^{ij}(x)$. The above proposition shows that if in addition the star product is alternative, then it is Hermitean. The Hermiticity of star products is essential for physical applications. 

Applying our procedure \cite{starpr} to construct the Weyl star product representing quantization of a Poisson bi-vector $\mathcal{P}^{ij}(x)$ one will not necessary end up with the associative star product. To make it associative one needs to introduce the corrections in $\alpha$ (renormalization) of the given bi-vector $\mathcal{P}^{ij}(x)$,  \begin{equation*}\mathcal{P}^{ij}\rightarrow \mathcal{P}_r^{ij}=\mathcal{P}^{ij}+\alpha^2\mathcal{P}^{ij}_2+\mathcal{O}\left( \alpha ^{4}\right),\end{equation*} see  \cite{KV} for details and for the construction of first nontrivial correction $\mathcal{P}^{ij}_2(x)$. For linear Poisson structure no corrections are needed and the Weyl star product constructed according to \cite{starpr} is associative.

The exemple considered in the previous section indicates that the similar situation takes place in case of the alternative star products representing quantization of linear Malcev-Poisson structure. 
In general, we expect that the renormalization of the Malcev-Poisson structure $P^{ij}(x)$ will make the corresponding Weyl star product alternative. 

\section{Trace functional}

In physics beside the star product one also needs the trace functional (integration) on the corresponding algebra of functions. The existence of the trace functional compatible with the associative star product is known as the Connes-Flato-Strenheimer conjecture. It was proven in \cite{FS}. Namely, it was shown that for any Poisson bi-vector $\mathcal{P}^{ij}(x)$ there exists a closed star product, i.e., the one satisfying
\begin{equation}\label{trace}
    \int f\star g=   \int  f\cdot g,
\end{equation}
where the integration is performed with a volume form ${\bf\Omega}$, such that $\mathrm{div}_{\bf\Omega}\mathcal{P}=0$. To find the star product satisfying (\ref{trace}) one may start with some known expression, e.g., the associative Weyl star product \cite{KV}, and then use the gauge freedom in its definition (\ref{gauge}) to obtain the desirable star product. The perturbative procedure of the construction of gauge operator $\mathcal{D}$ relating the Weyl star product and the closed one for a given Poisson structure $\mathcal{P}^{ij}(x)$ was proposed in \cite{Kup15}.

In the Sec. 4, we saw that the same idea \cite{Kup15} is also working in the case of the alternative star product. Using the gauge transformation (\ref{oct8}) it was constructed the closed alternative star product (\ref{oct7}) representing the algebra of imaginary octonions. Now let us consider  the star product corresponding to the arbitrary Malcev-Poisson structure $P^{ij}(x)$. To define the integration measure $\mu(x)\neq0$, we impose the condition
\begin{equation}\label{measure}
 \partial_i\left(\mu\cdot P^{ij}\right)=0 .\end{equation}
 In this case, the integrated bracket of two Schwartz functions vanishes, $\int \{f,g\}=0$, implying that 
 \begin{equation}\label{n2}
 \int  d^{N}x\mu [f\star g- f\cdot g]=\mathcal{O}\left( \alpha ^{2}\right) .\end{equation}
In open string theory one may choose, $\mu=\sqrt{\det(g+\mathcal{F})}$, as a Born-Infeld measure.In this case, (\ref{measure}) are the equations of motion for the corresponding gauge potencial  \cite{MK2, Blum}. The problem is that for the Weyl star product (\ref{starw}), already in the second order, $\mathcal{O}\left( \alpha ^{2}\right)$, the l.h.s. of (\ref{n2}) does not vanish. To give the precise expression of (\ref{n2}) we note that because of the identities,
\begin{equation*}\int d^{N}x\partial _{i}f\partial _{l}\left( \mu
\Pi ^{ilk}\right) \partial _{k}g=0,\,\,\,\,\partial _{l}\left( \mu
P ^{lj}\partial _{j}P ^{ki}\right) =0,\end{equation*}
one has,
\begin{equation*}
\int d^{N}x\partial _{i}f\partial _{l}\left( \mu
P ^{ij}\partial _{j}P ^{lk}\right) \partial _{k}g=\int d^{N}x\partial _{i}g\partial _{l}\left( \mu
P ^{ij}\partial _{j}P ^{lk}\right) \partial _{k}f,
\end{equation*}
i.e., the matrix 
\begin{equation}\label{n4}
b^{ik}(x)=\partial _{i}f\partial _{l}\left( \mu P ^{ij}\partial _{j}P ^{lk}\right) \partial _{k}g
\end{equation} 
is symmetric in the indices $i,k$ up to the surface terms. Now we find,
\begin{equation}\label{n5}
 \int  d^{N}x\mu [f\star g- f\cdot g] =-\frac{\alpha ^{2}}{6}\int d^{N}x\partial _{i}fb^{ik} \partial _{k}g+\mathcal{O}\left( \alpha ^{3}\right) .
\end{equation}
To fix it, we will follow \cite{Kup15} and look for the gauge transformation $\mathcal{D}:\star\rightarrow \circ$, mapping the given star product $\star$ to the closed one $\circ$. Introducing the gauge operator 
\begin{equation*}
\mathcal{D}=1-\frac{\alpha ^{2}}{12\mu }b^{ik}\partial _{i}\partial
_{k}+\mathcal{O}\left( \alpha ^{3}\right),
\end{equation*}
one obtains the star product
\begin{eqnarray}\label{circ}
f\circ g&=& \mathcal{D}^{-1}\left( \mathcal{D}f\star   \mathcal{D}g \right)\\
&=& f\star  g +\frac{\alpha ^{2}}{12\mu}\partial _{l}\left( \mu P ^{ij}\partial
_{j}P ^{lk}\right)\left[\partial _{i}f\partial
_{k}g+\partial _{i}g\partial_{k}f\right]+\mathcal{O}\left( \alpha ^{3}\right).\nonumber
\end{eqnarray}
which is closed up to the third order in $\alpha$:
\begin{equation*}
 \int  d^{N}x\mu [f\circ g- f\cdot g]=\mathcal{O}\left( \alpha ^{3}\right) ~.
\end{equation*}
Since, the new star product (\ref{circ}) is gauge equivalent to the alternative one (\ref{starw}), it is also alternative. Hence, due to (\ref{alt}) for this star product holds,
\begin{equation}\label{trace2}
    \int (f\circ g)\circ h=   \int  f\circ (g\circ h).
\end{equation}
{\it For the alternative closed star product the integrated associator vanishes.} This result is of great importance in physics, see e.g. \cite{MK2, Blum}. In particular, it means that on-shell, i.e. if the string equation of motion are satisfied, there will be no violation of associativity in string scattering amplitudes.

Another important property of the alternative star products is the Artin's theorem, which states that in an alternative algebra the subalgebra generated by any two elements is associative. That is, expressions (power products) involving only two functions, say $f$ and $g$, can be written unambiguously without parentheses. For the product of three elements, like $f\star g\star f$, this statement just reflects the alternative identities (\ref{a1}), (\ref{a2}) and (\ref{a3}). For the product of more then three elements the statement can be proved by induction. Consequently, alternative algebras are power-associative, that is, the subalgebra generated by a single element is associative. This means that if a function $f$ is multiplied by itself several times, it doesn't matter in which order the multiplications are carried out, so for instance
\begin{equation*}
 f\star(f\star(f\star f)) = (f\star(f\star f))\star f = (f\star f)\star(f\star f)=f\star f\star f\star f.
 \end{equation*}
These properties are essential for the consistent definition of the non-associative field theories. In particular, for the closed alternative star product using the Artin's theorem one may show that,
\begin{equation*}
\delta\left(\int V_\star(\Phi)\right)=\int\delta\Phi\cdot V^\prime_\star(\Phi),
\end{equation*}
where the symbol $V_\star(\Phi)$ means that in the function $V(\Phi)$ the ordinary multiplication was substituted by the star products. Variation of the interaction term is calculated in the same manner as in the standard non-commutative theory.

\section{Conclusions}

In this manuscript we have shown that the requirement of the associativity of star products can be relaxed in such a way that new star product will satisfy an important for physical applications properties, like the Moufang identities, alternative identities, etc. We substitute the condition of the associativity by the requirement (\ref{alt}) meaning that the star multiplication should be alternative. The condition (\ref{alt}) is gauge invariant from the point of view of deformation quantization. We show that in this case an integrated associator vanishes, which can be used to show that the non-associativity will not manifest itself on the physical observables. 

At the same time, the condition (\ref{alt}) implies the restrictions on the classical bracket $\{f,g\}=P^{ij}(x)\partial_i f\partial_jg$, it should be a Malcev-Poisson bracket (\ref{m3}). We discuss physical meaning of this restriction and show that on the Malcev-Poisson manifold there is a weaker form of the classical Poisson theorem, which can be used for the construction of integrals of motion.

Finally we note that there are some different approaches in treating the non-associative systems within the framework of deformation quantization. For exemple, one may require the associativity of the star product only for the physical observables \cite{LSh}, or introduce new elements to the original algebra which will make it associative \cite{MSS1,Ho}. 

The parts of this work were presented on the conference Non-commutative field theory and quantum gravity, Corfu, Greece, 2015, and submitted to the proceedings of Corfu Summer Institute \cite{PoS}.

\section*{Acknowledgements}
I appreciate the valuable remarks of Ivan Shestakov and Richard Szabo, I am also grateful to Jim Stasheff for correspondence. This work was supported in part by FAPESP, projects 2014/03578-6 and 2016/04341-5, and CNPq, project 443436/2014-2.

\end{document}